\documentclass[journal,twoside,web]{ieeecolor}
\usepackage{lcsys}
\usepackage{cite}
\usepackage{amsmath,amssymb,amsfonts}
\usepackage{algorithmic}
\usepackage{graphicx}
\usepackage{textcomp}

\DeclareMathSymbol{:}{\mathord}{operators}{"3A}
\DeclareMathSymbol{\shortminus}{\mathbin}{AMSa}{"39}

\newcommand{\tb}{}

\newcommand{\real}{\mathbb{R}}

\newcommand{\natrl}{\mathbb{N}}
\newcommand{\indx}[1]{\mathcal{I}_{#1}}

\newcommand{\union}{\operatorname{\cup}}
\newcommand{\intersection}{\ensuremath{\operatorname{\cap}}}

\newcommand{\eig}[1]{\mathrm{eig}({#1})}
\newcommand{\blkdiag}[2]{\mathrm{blkdiag}({#1}, \ldots, {#2})}
\newcommand{\col}[2]{\mathrm{col}({#1}, \ldots, {#2})}

\newcommand{\schurcomp}[2]{{#1} / {#2}}
\newcommand{\transp}{{\sf T}}

\newcommand{\setdef}[2]{\{#1 \;|\; #2\}}

\newcommand{\define}{\triangleq}

\newtheorem{theorem}{Theorem}
\newtheorem{proposition}{Proposition}

\newtheorem{lemma}{Lemma}

\newtheorem{definition}{Definition}

\newtheorem{problem}{Problem}

\newenvironment{pfof}[1]{\vspace{1ex}\noindent{\itshape Proof of #1:}\hspace{0.5em}} {\hfill \hspace*{\fill}$\blacksquare$ \vspace{1ex}}

\def\BibTeX{{\rm B\kern-.05em{\sc i\kern-.025em b}\kern-.08em
    T\kern-.1667em\lower.7ex\hbox{E}\kern-.125emX}}
\begin{document}

\title{A Lyapunov Characterization of Robust D-Stability with Application to Decentralized Integral Control of LTI Systems}

\author{John-Paolo Casasanta, \IEEEmembership{Student Member, IEEE}, and John W. Simpson-Porco, \IEEEmembership{Senior Member, IEEE}
\thanks{Research supported by NSERC Discovery RGPIN-2024-05523 and the NSERC CGS-M Award}
\thanks{The authors are with the Department of Electrical and Computer Engineering, University of Toronto, 10 King’s College Road, Toronto, ON, M5S 3G4, Canada (email: jp.casasanta@mail.utoronto.ca; jwsimpson@ece.utoronto.ca)}
\thanks{}
\thanks{}}

\maketitle

\begin{abstract}
    The concept of matrix D-stability plays an important role in applications, ranging from economic and biological system models to decentralized control. Here we provide necessary and sufficient Lyapunov-type conditions for the robust (block) D-stability property. We leverage this characterization as part of a novel Lyapunov analysis of decentralized integral control for MIMO LTI systems, providing sufficient conditions guaranteeing stability under low-gain and under arbitrary connection and disconnection of individual control loops.
\end{abstract}


%------------------------------------
%------------------------------------
%------------------------------------

\section{Introduction}
\label{sec:introduction}

A square matrix \(A\) is \emph{\(D\)-stable} if \(DA\) is Hurwitz for every positive diagonal matrix \(D\). This property has been studied for many decades, dating at least to the work of Enthoven and Arrow in economics \cite{ACE-KJA:56}, and subsequently developed a substantial literature \cite{KJA-MM:58, CRJ:74, CRJ:75, BEC:76}. More generally, one may consider \emph{robust \(D\)-stability}, which requires that \(D\)-stability persist under arbitrarily small perturbations of \(A\); see \cite{DJH:80, BEC:84, EA:86, WSK:02}. Both notions admit natural block variants, in which \(A\) and the admissible scalings \(D\) are required to respect a prescribed block structure. In this paper, we work in this block setting. 

A major motivation for studying \(D\)-stability and its variants is their appearance in applications. These notions arise naturally in economics \cite{ACE-KJA:56, MUR:25}, in the analysis of biological systems \cite{FB-EF-GG:12}, process control \cite{PJC-MM:94, SWS-AVS-YG-BGC-HTN:15, SWS-ZZ-BC-AS:25}, 
% the control of networked systems \cite{SRG-DD:15}, 
and in multi-parameter singular perturbation theory \cite{HKK-PVK:79, EA:86}.

Despite a long history, useful characterizations of (robust, block) \(D\)-stability remain difficult to obtain. 
{\tb In particular, the most immediate Lyapunov characterization of robust block $D$-stability (see Proposition \ref{Prop:RobDStabLyap}) is both computationally intensive to verify, and cumbersome to use in Lyapunov analysis.}
{\tb Our particular motivating need for an improved Lyapunov characterization arises from the study of decentralized integral controllers. It is well known that integral controllers can be used to asymptotically track constant references and reject constant disturbances --- the key idea is to integrate the tracking error $e \in \real^p$ and subsequently design an additional stabilizing compensator to guarantee closed-loop stability \cite{BAF-WMW:76}. In some practical applications (e.g., automatic generation control of power systems), the system to be controlled is already stable or has been stabilized, and a simplified integral control design of the form
\[
\dot{\eta} = -\varepsilon e, \qquad u = K\eta
\]
can be appended as a secondary control loop to improve tracking performance, where the matrix $K$ is a control gain to be designed offline, and the parameter $\varepsilon > 0$ is an integral gain to be tuned online. Typically, $\varepsilon$ must be kept small, in order to maintain closed-loop stability. Such designs are therefore known as \emph{low-gain integral controllers} (see \cite{ED:76b, CD-CAL:85, JWSP:21, PL-GW:23, PL-MG-RP-DN:25}). Extending this to a decentralized control setting where $N$ interacting subsystems must be independently controlled, each sub-controller will be a separate low-gain integral controller with gain $\varepsilon_i$ (see \eqref{eq:ControllersExpanded} in Section \ref{sec:intControl}, and additionally see \cite{ED:76a, PJC-MM:94, SWS-JB-PLL:04}). In such a decentralized configuration, $D$-stability requirements occur naturally, as each local integral gain $\varepsilon_i > 0$ must be tuned independently.

A related application motivating this study is \emph{feedback-based optimization}, wherein integral control-based optimization algorithms are implemented online to drive a dynamic system towards an optimal point of operation (see, e.g., \cite{AH-ZH-SB-GH-DF:24}). Decentralized low-gain integral control techniques are similar to those used during the analysis of decentralized and game-theoretic feedback-based optimization architectures (see \cite{AA-JWSP-LP:23, WW-ZH-GB-SB-FD:24, GB-DLM-MHB-SB-RSS-JL-FD:25}), and this will be a topic of future development.

}

\smallskip

\textit{Contributions:} We provide a new Lyapunov-based characterization of robust block \(D\)-stability. The result states that robust block \(D\)-stability is equivalent to a uniform boundedness property for the solutions of an associated Lyapunov equation. We then use this result to analyze decentralized low-gain integral controllers for LTI systems. In particular, we provide a novel Lyapunov analysis establishing that if \(-GK\) is robustly block \(D\)-stable, where \(G\) denotes the system DC gain and \(K\) is a block diagonal integral gain design, then the closed-loop system remains stable under arbitrary connection and disconnection of integral loops, as well as under arbitrary block-wise detuning of the individual integral gains. We illustrate the result with a simple simulation.

\smallskip

\emph{Notation:} {\tb Throughout $\| \cdot \| \define \| \cdot \|_2$ denotes the induced matrix $2$-norm, with $\| \cdot \|_1$ the induced $1$-norm}. For $N \in \natrl$, $\indx{N} \define \{1, \ldots, N\}$ denotes an index set. With an abuse of notation, let $\alpha \subset \indx{N}$ denote \emph{ordered} index sets of cardinality $n_\alpha = |\alpha|$, with $\alpha(i)$ the $i$\textsuperscript{th} element of $\alpha$. When no ordering is explicitly prescribed, the standard ordering on $\natrl$ is assumed. For a block partitioned matrix {\tb $A$,}
where $A_{i,j}$ are of appropriate sizes, and two ordered index sets $\alpha,\beta \subset \indx{N}$, we denote by
\begin{equation} \label{eq:PartMatSub}
    A_{\alpha, \beta} \define
    \left[
    \begin{array}{c|c|c}
    A_{\alpha(1), \beta(1)} & \cdots &  A_{\alpha(1), \beta(n_\beta)} \\ \hline
    \vdots  & & \vdots \\ \hline
    A_{\alpha(n_\alpha), \beta(1)} &  \cdots & A_{\alpha(n_\alpha), \beta(n_\beta)}
    \end{array}
    \right]
\end{equation}
the submatrix formed from the blocks prescribed in $\alpha$ and $\beta$. Similarly, for a block row, block column, or block diagonal matrix $A$, we denote by $A_{\alpha}$ the submatrix formed from the blocks prescribed in $\alpha$. Finally, let $A$ be a square block partitioned matrix with square diagonal blocks, and let $\alpha \subset \indx{N}$ be an ordered index set. If $A_{\alpha,\alpha}$ is non-singular, the Schur complement of $A$ relative to $A_{\alpha, \alpha}$ is denoted by
\begin{equation}
    \schurcomp{A}{A_{\alpha, \alpha}} \define A_{\Bar{\alpha}, \Bar{\alpha}} - A_{\Bar{\alpha}, \alpha} A_{\alpha, \alpha}^{-1} A_{\alpha, \Bar{\alpha}}.
\end{equation} 

%------------------------------------
%------------------------------------
%------------------------------------

\section{Review of Block $D$-Stability Concepts}
\label{sec:DStab}

We briefly recall the key concepts and results of (robust) block $D$-stability; a thorough review of $D$-stability concepts can be found in \cite{OYK:19}. Consider {\tb a} block partitioned matrix {\tb $A$}
with blocks $A_{i,j} \in \real^{n_i \times n_j}$. {\tb The partition sizes $\{ n_i \}_{i=1}^N$ are fixed.}
\begin{equation*}
    \mathcal{D}_{\rm blk} = \left\{ \blkdiag{d_1 I_{n_1}}{d_N I_{n_N}} \,\,  \Big| \,\, d_1, \ldots, d_N > 0 \right\}.    
\end{equation*}

\begin{definition}\label{Def:DStab}
A block partitioned square matrix $A$ is \emph{block $D$-stable} if $DA$ is Hurwitz for every $D \in \mathcal{D}_\mathrm{blk}$.
\end{definition}

\smallskip

Block $D$-stability requires that all matrices generated by positive scaling of every block row of $A$ be Hurwitz. The following Lyapunov characterization is immediate from standard results (e.g.,\cite[Proposition 11.9.5]{DSB:11}).

\smallskip

\begin{proposition} \label{Prop:DStabLyap}
    A block partitioned square matrix $A$ is block $D$-stable if and only if for any ${\tb Q = Q^\transp > 0}$ and any $D \in \mathcal{D}_\mathrm{blk}$, there exists a unique ${\tb P_D = P_D^\transp > 0}$ such that 
    \begin{equation}\label{eq:DStabLyap}
        A^\transp D P_D + P_D D A = -Q.
    \end{equation}
\end{proposition}

\smallskip

A matrix which is block $D$-stable must of course be Hurwitz (take $D = I$), and in the case where there is only a single partition ($N = 1$), the property is clearly equivalent to Hurwitz stability. Block \emph{diagonal} stability of $A$ \textemdash{} the existence of a block diagonal ${\tb P = P^\transp > 0}$ such that ${\tb A^{\sf T}P + PA < 0}$ \textemdash{} is a common sufficient (but \emph{not} necessary) condition for block $D$-stability. By extension, this implies stable diagonally dominant matrices and stable Metzler matrices are block $D$-stable \cite[Chapter 6]{AB-RJP:94}. Verifying $D$-stability is in general difficult; see \cite{CRJ:74} for several additional sufficient conditions.

The $D$-stability property may unfortunately be lost under arbitrarily small perturbations to the matrix elements. For instance, consider the partitioned matrix 
\begin{equation} \label{eq:DStabEdgeEx}
    A_{\varepsilon} \define
    \left[
    \begin{array}{c|c}
    \varepsilon & 1 \\ \hline
    -1 & -2 
    \end{array}
    \right]
\end{equation}
where $\varepsilon \in \real$. One may verify that $A_0$ is $D$-stable, but that for any arbitrarily small $\varepsilon > 0$, $A_{\varepsilon}$ is not $D$-stable \cite{DJH:80}. This lack of robustness is addressed by defining \emph{robust block $D$-stability} (also known as \emph{strong block $D$-stability} \cite{EA:86}).

\smallskip

\begin{definition}\label{Def:RobDStab}
    A block partitioned square matrix $A$ is \emph{robustly block $D$-stable} if there exists $\mu > 0$ such that $A + B$ is block $D$-stable for every $B$ satisfying $\|B\| \leq \mu$.
\end{definition}

\smallskip

By definition, the set of robustly block $D$-stable matrices is the interior of the set of block $D$-stable matrices. Unlike block $D$-stability, robust block $D$-stability is inherited by block principal submatrices and block Schur complements.

\smallskip

\begin{proposition} \label{Prop:RobustSub}
    Suppose the partitioned square matrix $A$ is robustly block $D$-stable, and let $\alpha \subseteq \indx{N}$ be an ordered index set. Then, $A_{\alpha, \alpha}$ and $\schurcomp{A}{A_{\alpha, \alpha}}$ are robustly block $D$-stable.
\end{proposition}

\smallskip

{\tb The proof of Proposition \ref{Prop:RobustSub} is an immediate extension of \cite{CRJ:75} and \cite[Lemma 1]{JL-TFE:01} to the case of (robust) block $D$-stability.}
Most sufficient conditions for block $D$-stability are strong enough to imply robust block $D$-stability, including block diagonal stability (see \cite{WSK:02} for robust extensions of \cite{CRJ:74}). A Lyapunov characterization in the robust case can be obtained by extending Proposition \ref{Prop:DStabLyap} in the obvious fashion.

\smallskip

\begin{proposition} \label{Prop:RobDStabLyap}
    A block partitioned square matrix $A$ is robustly block $D$-stable if and only if there exists $\mu > 0$ such that for any ${\tb Q = Q^\transp > 0}$, any $D \in \mathcal{D}_\mathrm{blk}$, and any $B \in \real^{n \times n}$ such that $\| B \| \leq \mu $, there exists a unique ${\tb P_{D,B} = P_{D,B}^\transp > 0}$ such that 
    \begin{equation}\label{eq:RobDStabLyap}
        (A+B)^\transp D P_{D,B} + P_{D,B} D (A+B) = -Q.
    \end{equation}
\end{proposition}

\smallskip

Proposition \ref{Prop:RobDStabLyap} is unfortunately too cumbersome to apply, as it involves quantification over the perturbation matrix $B$, which ranges over the set $\setdef{B}{\|B\| \leq \mu}$. Moreover, the result provides little new insight into the robust $D$-stability property, and what separates it from mere $D$-stability. We next develop a much simpler Lyapunov characterization of robust block $D$-stability which eliminates the quantification over the perturbation matrix.
 
%------------------------------------
%------------------------------------
%------------------------------------

\section{Main Result: A Lyapunov Characterization of Robust Block $D$-Stability}
\label{sec:RobDStabLyap}

{\tb We now present a characterization of robust block $D$-stability which addresses the drawbacks of Proposition \ref{Prop:RobDStabLyap}.}

\smallskip

\begin{theorem} \label{Theorem:RobStabLyap}
A block partitioned square matrix $A$ is robustly block $D$-stable if and only if for any {\tb fixed} ${\tb Q = Q^\transp > 0}$, there exists $M > 0$ such that for any $D \in \mathcal{D}_\mathrm{blk}$
\begin{enumerate}
    \item[(i)] there exists a unique ${\tb P_D = P_D^\transp > 0}$ satisfying \eqref{eq:DStabLyap};
    \item[(ii)] $\| P_D D \| \leq M$.
\end{enumerate}
\end{theorem}

\medskip

Comparing Theorem \ref{Theorem:RobStabLyap} with Proposition \ref{Prop:DStabLyap}, robust block $D$-stability is equivalent to block $D$-stability with a \emph{uniform} bound on the Lyapunov quantity $\| P_D D \|$.
{\tb Theorem \ref{Theorem:RobStabLyap} has several key advantages over Proposition \ref{Prop:RobDStabLyap}. First, the uniform bound $\| P_D D \| \leq M$ provides new intuition on how solutions of the Lyapunov equation scale with $D$. Second, the bound provides significant utility within Lyapunov or small-gain type arguments, as will be demonstrated in the proof of Theorem \ref{Theorem:ProbSol}. 
}
{\tb Further, assuming the existence of a unique solution $P_D$ to \eqref{eq:DStabLyap} which is parameterized by $N$ variables $(d_1, \ldots, d_N)$, one may check positive definiteness of $P_D$ and boundedness of $\| P_D D \| $ to verify robust block $D$-stability. Such an approach is computationally viable for small $N$; in contrast, a similar approach leveraging Proposition \ref{Prop:RobDStabLyap} is much more computationally intensive, as $P_{D,B}$ would need to be parameterized by $N$ variables $(d_1, \ldots, d_N)$ plus $n^2$ variables corresponding to $B$.}
{\tb The authors note the development of computational methods for verifying (and designing) $D$-stable matrices appears to be a mostly open area of research.}

The uniform bound $\| P_D D \| \leq M$ fails for merely $D$-stable matrices: consider the matrix $A_\varepsilon$ defined in \eqref{eq:DStabEdgeEx}, which is robustly $D$-stable for $\varepsilon < 0$ and $D$-stable for $\varepsilon \leq 0$. Selecting $Q = I$, given arbitrary ${\tb D = \mathrm{diag}(d_1,d_2) > 0}$, the unique solution to \eqref{eq:DStabLyap} is
\[
    P_D = \frac{1}{ f(d_1, d_2, \varepsilon)} \begin{bmatrix} \frac{d_1(1 - 2 \varepsilon) + 5 d_2}{d_1} & (\varepsilon + 2) \\ (\varepsilon + 2) & \frac{d_1(\varepsilon^2 + 1) + d_2(1 - 2\varepsilon)}{d_2} \end{bmatrix},
\]
where we have defined $f(d_1, d_2, \varepsilon) \define d_1\varepsilon(4\varepsilon - 2) + d_2(4 - 8\varepsilon)$. Setting $\varepsilon = 0$ and $d_2 = 1$, one can calculate {\tb$ \| P_D D \| \geq 2^{-1/2}\| P_D D \|_1 = (3 d_1 + 5)/2^{5/2}$ \cite[Fact 9.8.12]{DSB:11}}, which is unbounded as $d_1 \rightarrow \infty$. Thus, an upper bound on $\| P_D D \|$ which is uniform for all $D \in \mathcal{D}_\mathrm{blk}$ cannot be established for $A_0$. We contrast the case where $\varepsilon = -1$, where one can calculate {\tb $\| P_D D \| \leq \| P_D D \|_1 = (4d_1 + 5d_2) / (6d_1 + 12d_2)$ }, a quantity which is obviously bounded over all $d_1, d_2 > 0$. We now continue with the proof.

\smallskip

\smallskip
\begin{proof}
We first show the conditions given in Theorem \ref{Theorem:RobStabLyap} imply robust $D$-stability of $A$. Let $\mu \in (0,\lambda_{\rm min}(Q) / 2M)$, and consider any perturbation $B \in \real^{n \times n}$ such that $\| B \| \leq \mu$. For $D \in \mathcal{D}_\mathrm{blk}$, ${\tb P_D = P_D^\transp > 0}$ satisfying \eqref{eq:DStabLyap}, we calculate that
\begin{equation} \label{Eq:DStabLyapPert}
    (A+B)^\transp D P_D + P_D D (A+B) = -\tilde{Q},
\end{equation}
where we have defined
\begin{equation*}
    \tilde{Q} \define Q - (B^\transp D P_D + P_D D B).
\end{equation*}
Standard eigenvalue bounds \cite[Lemma 8.4.1]{DSB:11} imply that
{\tb
\begin{equation*}
\begin{aligned}
    B^\transp D P_D + P_D D B &\leq \lambda_\mathrm{max}(B^\transp D P_D + P_D D B) I \\
    &\leq 2 \| P_D D \| \| B \| I \\
    &\leq 2 M \mu I < \lambda_\mathrm{min}(Q) I,
\end{aligned}
\end{equation*}}
and thus ${\tb \tilde{Q} > 0}$. Thus, since $\tilde{Q}, {\tb P_D > 0}$ in \eqref{Eq:DStabLyapPert}, Lyapunov's theorem implies that $D(A+B)$ is Hurwitz. As this argument holds for any $D \in \mathcal{D}_\mathrm{blk}$, $B$ such that $\| B \| \leq \mu$, $A$ is robustly block $D$-stable.

\smallskip

We now show robust block $D$-stability of $A$ implies the conditions given in Theorem \ref{Theorem:RobStabLyap}. As $A$ is assumed robustly block $D$-stable, it must be block $D$-stable trivially. By Proposition \ref{Prop:DStabLyap}, for any ${\tb Q > 0}$, given any $D \in \mathcal{D}_\mathrm{blk}$, there exists ${\tb P_D > 0}$ such that \eqref{eq:DStabLyap} is satisfied. We must show there exists $M > 0$ such that $\| P_D D \| \leq M$ uniformly in $D \in \mathcal{D}_{\rm blk}$.
The proof proceeds by contradiction. Suppose that $A$ is robustly $D$-stable, yet there exists a sequence $\{ D^k \}_{k=1}^{\infty} \subset \mathcal{D}_{\rm blk}$ and corresponding sequence $\{ P_{D^k} \}_{k=1}^{\infty}$ determined by \eqref{eq:DStabLyap} such that $\{\|P_{D^k}D^k\| \}_{k=1}^{\infty} \subset \real$ is unbounded. 
We express the sequence $\{D^k\}_{k=1}^\infty$ in a helpful form using the following lemma, which is proved in Appendix \ref{Appendix:SupportProofs}.

\smallskip

\begin{lemma} \label{Lemma:BlockDSeq}
    Consider any sequence $\{ D^k\}_{k=1}^\infty \subset \mathcal{D}_\mathrm{blk}$. Then $D^k$ can be expressed in the form {\tb $D^k = \Sigma^\transp \Tilde{D}^k T^k \Sigma$,}
where
\begin{enumerate}
\item[(i)] $\Sigma$ is a (block) permutation matrix;
\item[(ii)] $\Bar{N} \leq N$ is a positive integer;
\item[(iii)] $\Tilde{D}^k = \blkdiag{\tilde{D}_1^k}{\tilde{D}_{\bar{N}}^k} \in \mathcal{C}$, where $\mathcal{C} \subset \mathcal{D}_{\rm blk}$ is a compact subset;
\item[(iv)] $T^k = \blkdiag{\tau_1^k I}{\tau_{\bar{N}}^k I} \in \mathcal{D}_\mathrm{blk}$ satisfying $\lim_{k \rightarrow \infty} \tau_i^k / \tau_j^k = 0$ if and only if $i > j$. 
\end{enumerate}
\end{lemma}

\smallskip

The idea is that {\tb Lemma \ref{Lemma:BlockDSeq}}
organizes the $\bar{N}$ relative rates of growth or decay of the blocks of $D^k$.
{\tb Going forward, a variable substitution enabled by Lemma \ref{Lemma:BlockDSeq} will facilitate the derivation of a Sylvester equation from \eqref{eq:DStabLyap}. The solution of this equation at the limit $k \rightarrow \infty$ will lead to a contradiction which will complete the proof. We now proceed.}
As $\{\Tilde{D}^k\}_{k=1}^\infty$ lies in the compact set $\mathcal{C}$, it admits a convergent subsequence in $\mathcal{C}$ \cite[Theorem 3.6]{WR:76}; we relabel the sequence without loss of generality, and denote by $\tilde{D}^{\star} \define \lim_{k \rightarrow \infty} \tilde{D}^k \in \mathcal{D}_\mathrm{blk}$ the limiting value of the new sequence $\{\tilde{D}^k\}_{k=1}^{\infty}$.
Returning now to \eqref{eq:DStabLyap}, define $X^k \define P_{D^k} D^k / \| P_{D^k}D^k \|$. By construction, $\| X^k \| = 1$ and by symmetry of $P_{D^k}$ and $D^k$, it holds that 
\begin{equation}\label{Eq:XkSymmetry}
D^k X^k = {X^k}^\transp D^k.
\end{equation}
Since $\|X^k\| = 1$, the sequence $\{ X^k \}_{k=1}^{\infty}$ lies in a compact set, and thus also admits a convergent subsequence; relabelling the sequence again without loss of generality, we denote by $X^{\star} \define \lim_{k \rightarrow \infty} X^k$ the limiting value. 
{\tb Dividing though the Lyapunov equation \eqref{eq:DStabLyap} by $\| P_{D^k}D^k \|$, substituting the definition of $X^k$, and taking the norm of both sides of the equation, we obtain}
\begin{equation}\label{Eq:PDNormDim}
    \left \| A^\transp {X^k}^\transp +  X^k A \right \| = \varepsilon_k,
\end{equation}
where $\varepsilon_k = \|Q\| / \| P_{D^k}D^k \| > 0$ converges to zero. 
Consider now the coordinate transformation
\[
Y^k \define \Tilde{D}^k \Sigma X^k \Sigma^\transp \qquad \Longleftrightarrow \qquad X^k = \Sigma^{\transp}(\tilde{D}^k)^{-1}Y^k\Sigma.
\]
It follows from \eqref{Eq:XkSymmetry} that $\{ Y^k \}_{k=1}^{\infty}$ satisfies $T^k Y^k = {Y^k}^\transp T^k$.
Let $Y^{\star} = \Tilde{D}^{\star} \Sigma X^{\star} \Sigma^\transp$ denote the limit of the sequence $\{ Y^k \}_{k=1}^{\infty}$. Block by block, the equality $T^k Y^k = {Y^k}^\transp T^k$ states that
\begin{equation}
    \tau_i^k Y_{i,j}^k = {Y_{j,i}^k}^\transp \tau_j^k, \qquad i,j \in \indx{\bar{N}}.
\end{equation}
For $i=j$, we immediately see that $Y_{i,i}^k = {Y_{i,i}^k}^\transp$, implying $Y_{i,i}^{\star} = {Y_{i,i}^{\star}}^\transp$. For $j > i$, $ \lim_{k \rightarrow \infty} \tau_j^k / \tau_i^k = 0$, implying $Y^{\star}_{i,j} = 0$. Thus, $Y^{\star}$ is block lower triangular with symmetric matrices as its diagonal blocks. 
{\tb Substituting for $X^k$ in terms of $Y^k$, \eqref{Eq:PDNormDim}  now reads as}
\begin{equation} \label{Eq:PDEqEps}
\begin{aligned}
    &\|A^\transp \Sigma^\transp {Y^k}^\transp (\Tilde{D}^k)^{-1} \Sigma + \Sigma^\transp (\Tilde{D}^k)^{-1} Y^k \Sigma A\| = \varepsilon_k.
\end{aligned}
\end{equation}
The sequence inside the norm in \eqref{Eq:PDEqEps} converges, and thus
\begin{equation*}
A^\transp \Sigma^\transp {Y^{\star}}^\transp (\Tilde{D}^{\star})^{-1} \Sigma + \Sigma^\transp (\Tilde{D}^{\star})^{-1} Y^{\star} \Sigma A = 0
\end{equation*}
which, rearranging, implies
\begin{equation}\label{Eq:PDEq0}
 \bar{A}^\transp (Y^{\star})^\transp  + (Y^{\star}) \bar{A} = 0
\end{equation}
where $\Bar{A} \define \Sigma A \Sigma^\transp \Tilde{D}^{\star}$. Since $A$  is robustly block $D$-stable, so is $\bar{A}$ \cite[Observation 2]{CRJ:74}. We now decompose $\Bar{A}$, $Y^{\star}$ such that 
\begin{equation} \label{Eq:PDABarYStarExp}
\begin{aligned}
    \Bar{A} = \begin{bmatrix} \Bar{A}_{1,1} & \Bar{A}_{1, \alpha} \\ \Bar{A}_{\alpha, 1} & \Bar{A}_{\alpha, \alpha} \end{bmatrix},\,\,
\end{aligned}
\begin{aligned}
    Y^{\star} = \begin{bmatrix} Y^{\star}_{1,1} & 0 \\ Y^{\star}_{\alpha, 1} & Y^{\star}_{\alpha, \alpha} \end{bmatrix},
\end{aligned}
\end{equation}
where $\alpha \define \indx{{\tb \Bar{N}} } \setminus \{1\}$. From the previous conclusion, $Y_{1,1}^{\star}$ is symmetric, and $Y^{\star}_{\alpha, \alpha}$ is lower block triangular with symmetric block diagonals. Substituting \eqref{Eq:PDABarYStarExp} into \eqref{Eq:PDEq0}, we obtain the three equations
\begin{subequations}\label{Eq:PDABarYStar}
\begin{align}
    \Bar{A}_{1,1}^\transp Y^{\star}_{1,1} + Y^{\star}_{1,1} \Bar{A}_{1,1} &= 0, \label{Eq:PDABarYStarEq1} \\ 
    \Bar{A}_{1, \alpha}^\transp Y^{\star}_{1,1} + Y^{\star}_{\alpha, 1} \Bar{A}_{1,1} + Y^{\star}_{\alpha, \alpha} \Bar{A}_{\alpha, 1} &= 0, \label{Eq:PDABarYStarEq2} \\ 
    (Y^{\star}_{\alpha, 1} \Bar{A}_{1, \alpha} + Y^{\star}_{\alpha, \alpha} \Bar{A}_{\alpha, \alpha}) + (*)^\transp &= 0. \label{Eq:PDABarYStarEq3}
\end{align}
\end{subequations}
The first equation \eqref{Eq:PDABarYStarEq1} is a Sylvester equation, which will have a unique solution if and only if $\eig{\Bar{A}_{1,1}} \cap \eig{-\Bar{A}_{1,1}} = \emptyset$ \cite[Proposition 7.2.3 and 7.2.4]{DSB:11}. Since $\Bar{A}$ is robustly block $D$-stable, $\Bar{A}_{1,1}$ is Hurwitz by Proposition \ref{Prop:RobustSub}, trivially implying $\eig{\Bar{A}_{1,1}} \cap \eig{-\Bar{A}_{1,1}} = \emptyset$. Thus, \eqref{Eq:PDABarYStarEq1} has a unique solution, which by linearity is $Y^{\star}_{1,1} = 0$. Substituting $Y^{\star}_{1,1} = 0$ into \eqref{Eq:PDABarYStarEq2} and rearranging, we find
\begin{equation} \label{Eq:PDABarYStarEq2Simp}
    Y^{\star}_{\alpha, 1} = - Y^{\star}_{\alpha, \alpha} \Bar{A}_{\alpha, 1} \Bar{A}_{1,1}^{-1}
\end{equation}
which upon substitution into \eqref{Eq:PDABarYStarEq3} yields
\begin{equation} \label{Eq:PDRedMatEq0}
\begin{aligned}
    0 &= (\schurcomp{\Bar{A}}{\Bar{A}_{1,1}})^\transp (Y^{\star}_{\alpha, \alpha})^\transp + Y^{\star}_{\alpha, \alpha} (\schurcomp{\Bar{A}}{\Bar{A}_{1,1}}).
\end{aligned}
\end{equation}
By Proposition \ref{Prop:RobustSub}, $\schurcomp{\Bar{A}}{\Bar{A}_{1,1}}$ is robustly block $D$-stable. Thus, \eqref{Eq:PDRedMatEq0} is of the same form as \eqref{Eq:PDEq0}. By the line of argument below \eqref{Eq:PDABarYStar}, we recursively find that $Y^{\star}_{2,2} = \ldots = Y^{\star}_{\Bar{N},\Bar{N}} = 0$, and back substitution into \eqref{Eq:PDABarYStarEq2Simp} implies that $Y^* = 0$. This directly implies $X^{\star} = 0$, which contradicts that $\| X^{\star} \| = 1$. It follows that $\sup_{k \in \mathbb{N}}\|P_{D^k}D^{k}\| < \infty$ for every sequence $\{D^k\}_{k=1}^{\infty} \subset \mathcal{D}_{\rm blk}$, and thus $\sup_{D \in \mathcal{D}_{\rm blk}} \|P_{D}D\| < \infty$, which completes the proof. \end{proof}

%------------------------------------
%------------------------------------
%------------------------------------

\section{Application to Decentralized Integral Control of LTI Systems}
\label{sec:intControl}

Consider the MIMO LTI system with inputs and outputs partitioned into $N \in \natrl$ collections,
\begin{equation} \label{eq:LTIExpanded}
\begin{aligned}
    \Dot{x} &= Ax + \sum_{j = 1}^N\nolimits B^u_j u_j + B^w w, & \quad & x(0) = x_0, \\
    e_i &= C_i x + \sum_{j = 1}^N\nolimits D_{i,j}^u u_j + D_{i}^w w, & \quad & i \in \indx{N},
\end{aligned}
\end{equation}
with state $x \in \real^n$, local control inputs $u_i \in \real^{m_i}$, and local error signals $e_i \in \real^{p_i}$. The system is subject to a constant disturbance $w \in \real^{n_w}$. 
{\tb We assume $A$ is Hurwitz, meaning that either the open-loop system \eqref{eq:LTIExpanded} is inherently stable, or it has been stabilized.} Under this assumption, let
{\tb $G^u_{i,j} \define - C_i A^{-1} B^u_j + D_{i,j}^u$, $G^w_i \define - C_i A^{-1} B^w + D_i^w$}
denote the DC gain matrices of \eqref{eq:LTIExpanded} from $u_j$ to $e_i$ and $w$ to $e_i$, respectively. We denote $G^u, G^w$ the appropriate concatenations of the DC gain matrices from $u$ to $e$ and $w$ to $e$. We consider a set of local low-gain integral controllers
\begin{equation}\label{eq:ControllersExpanded}
    \Dot{\eta}_i = -\varepsilon_i e_i, \qquad u_i = K_i \eta_i, \qquad i \in \indx{N},
\end{equation}
where $\varepsilon_i > 0$ are tuning parameters, which process the errors $e_i$ to produce the control signals $u_i$. We introduce the compact notation $\eta = \col{\eta_1}{\eta_N}$, $\varepsilon = \col{\varepsilon_1}{\varepsilon_N}$, $K = \blkdiag{K_1}{K_N}$, $\mathcal{E} = \blkdiag{\varepsilon_1 I_{p_1}}{\varepsilon_N I_{p_N}}$. Before continuing, we remind the reader of the notation presented {\tb in \eqref{eq:PartMatSub}}. An ordered index set $\sigma \subset \indx{N}$ will be used to identify which integral control loops are closed. Specifically, if $i \in \sigma$, then the $i^\mathrm{th}$ integral control loop has been closed. We appropriately denote $\Bar{\sigma} = \indx{N} \setminus \sigma$ the set of control loops that have not been closed. This leads us to consider the set of closed loop systems $\Sigma^\sigma_\mathrm{cl}(\varepsilon_\sigma)$, defined by
\begin{equation}\label{eq:ClosedLoopSigma}
\begin{aligned}
    \Dot{x} &= Ax + B_\sigma^u u_\sigma + B_{\Bar{\sigma}}^u u_{\Bar{\sigma}} + B^w w, \\
    \dot{\eta}_\sigma &= -\mathcal{E}_\sigma e_\sigma, \\
    e_\sigma &= C_\sigma x + D_{\sigma, \sigma}^u u_\sigma + D_{\sigma, \Bar{\sigma}}^u u_{\Bar{\sigma}} + D_\sigma^w w 
\end{aligned}
\qquad
\begin{aligned}
    &u_{\sigma} = K_\sigma \eta_\sigma, \\
    &u_{\Bar{\sigma}} \in \real^{m_{\bar{\sigma}}}, \\
    &w \in \real^{n_w},
\end{aligned}
\end{equation}
where the closed and open loop control inputs have been collected into $u_{\sigma}$ and $u_{\bar{\sigma}}$, the relevant error signals have been collected in $e_\sigma$, and $m_{\Bar{\sigma}} \define \sum_{i \in \Bar{\sigma}} m_i$. 
{\tb The open loop control inputs $u_{\bar{\sigma}}$ are assumed to be arbitrary constants. Such an assumption enables the treatment of the open loop inputs $u_{\bar{\sigma}}$ as if they were constant disturbances.}
We study the following analysis problem.

\smallskip 

\begin{problem} \label{Prob:1}
Given a set of gains $\{K_i\}_{i=1}^{N}$, find conditions under which there exists an upper gain value $\varepsilon^{\star} > 0$ such that for each closed loop configuration $\sigma \subset \indx{N}$, each unused input $u_{\Bar{\sigma}} \in \real^{m_{\Bar{\sigma}}}$, and each  disturbance $w \in \real^{n_w}$, the closed-loop system $\Sigma^\sigma_\mathrm{cl}$ defined in \eqref{eq:ClosedLoopSigma} possesses a unique equilibrium point which is exponentially stable for all $\varepsilon_{\sigma} \in (0,\varepsilon^{\star})^{|\sigma|}$. 
\end{problem}

\smallskip

If Problem \ref{Prob:1} is feasible, the resulting controllers enjoy several strong properties, namely (i) they are local (ii) they may be implemented (or disconnected) asynchronously, and (iii) they may be tuned independently for desired local performance (as long as $\varepsilon_i \in (0,\varepsilon^*)$). These are all desirable properties for control of large-scale networked systems.

\subsection{Main Stability Result}

{\tb We now present a sufficient condition guaranteeing the feasibility of Problem \ref{Prob:1}}.

\begin{theorem}\label{Theorem:ProbSol} 
Assume that the open-loop system \eqref{eq:LTIExpanded} is exponentially stable. Then Problem \ref{Prob:1} is solvable if $-G^u K$ is robustly block $D$-stable.
\end{theorem}

\smallskip

In the case where we have a single partition ($N = 1$), Theorem \ref{Theorem:ProbSol} reduces to $-G^u K$ Hurwitz, as implicitly given for the case of centralized low-gain integral control in \cite{ED:76b}. Theorem \ref{Theorem:ProbSol} places conditions only on the gain matrices $K = \mathrm{blkdiag}(K_1,\ldots,K_N)$ and the system DC gain $G^u$, and thus relies on very minimal system model information. 

{\tb Analysis problems similar (but not identical) to Problem \ref{Prob:1} have been formalized in the literature (for example \cite{EA:86, PJC-MM:94}), with the corresponding stability analyses relying on triangularization and frequency domain techniques, respectively. In contrast, Theorem \ref{Theorem:RobStabLyap} enables us to perform the first {\tb \emph{Lyapunov-based}} stability analysis, similar in style to Lyapunov stability proofs using singular perturbation methods (see, e.g., \cite[Chapter 7.2]{PK-HKK-JOR:99}) in the centralized case \cite{JWSP:21}. The robust $D$-stability condition in Theorem \ref{Theorem:ProbSol} is of equal strength to conditions one would have obtained using the techniques given in \cite{EA:86, PJC-MM:94}, but with the added potential for extension to nonlinear systems that a Lyapunov framework provides.}

\smallskip

\begin{proof}
Consider any closed loop configuration $\sigma \subset \indx{N}$.
One may easily verify that
\begin{equation} \label{eq:EqbmXEta}
\begin{aligned}
    \bar{x}^\sigma &= -A^{-1}B^u_\sigma K_\sigma \bar{\eta}^\sigma_\sigma - A^{-1}B^u_{\Bar{\sigma}} u_{\Bar{\sigma}} - A^{-1}B^w w, \\
    \bar{\eta}_\sigma^\sigma &= (-G^u_{\sigma, \sigma} K_\sigma)^{-1} G^u_{\sigma, \Bar{\sigma}} u_{\bar{\sigma}} + (-G^u_{\sigma, \sigma} K_\sigma)^{-1} G^w_\sigma w    
\end{aligned}
\end{equation}
is the (necessarily unique) equilibrium point of $\Sigma^\sigma_\mathrm{cl}$ described in \eqref{eq:ClosedLoopSigma}. We note that $A$ is Hurwitz by assumption. Robust block $D$-stability of $-G^u K$ implies robust block $D$-stability of $-G^u_{\sigma, \sigma} K_\sigma$, where we have leveraged Proposition \ref{Prop:RobustSub}; this implies $-G^u_{\sigma, \sigma} K_\sigma$ must necessarily be Hurwitz. As $A$ and $-G^u_{\sigma, \sigma} K_\sigma$ are Hurwitz, their inverses in \eqref{eq:EqbmXEta} are well defined.
To show stability, consider the change of state variable
\begin{equation} \label{eq:ClosedLoopSimilarityTransform}
    \xi \define x - \Pi^u_\sigma K_\sigma \eta_\sigma - \Pi^u_{\Bar{\sigma}} u_{\Bar{\sigma}} - \Pi^w w,
\end{equation}
where we have defined
\begin{equation}
\begin{aligned}
    \Pi^u \define -A^{-1}B^u,
\end{aligned}
\qquad 
\begin{aligned}
    \Pi^w \define -A^{-1}B^w.
\end{aligned}
\end{equation}
{\tb Differentiating \eqref{eq:ClosedLoopSimilarityTransform} and using  the dynamics \eqref{eq:ClosedLoopSigma}, tedious but routine calculations yield the new closed-loop system}
\begin{equation}\label{eq:ClosedLoopSigmaCOV}
\begin{aligned}
    \begin{bmatrix} \Dot{\xi} \\ \Dot{\eta}_\sigma \end{bmatrix} = \mathcal{A}^{\sigma} \begin{bmatrix} \xi \\ \eta_\sigma \end{bmatrix} + \mathcal{B}^{\sigma} \begin{bmatrix} u_{\Bar{\sigma}} \\ w \end{bmatrix},
\end{aligned}
\end{equation}
where we have defined
\begin{equation}
\begin{aligned}
    \mathcal{A}^\sigma &\define \begin{bmatrix} A + \Pi^u_\sigma K_\sigma \mathcal{E}_\sigma C_\sigma & \Pi^u_\sigma K_\sigma \mathcal{E}_\sigma G^u_{\sigma, \sigma} K_\sigma \\ -\mathcal{E}_\sigma C_\sigma & -\mathcal{E}_\sigma G^u_{\sigma, \sigma} K_\sigma \end{bmatrix} \\
    \mathcal{B}^\sigma &\define \begin{bmatrix} \Pi^u_\sigma K_\sigma \mathcal{E}_\sigma G^u_{\sigma, \Bar{\sigma}} & \Pi^u_\sigma K_\sigma \mathcal{E}_\sigma G^w_\sigma \\ -\mathcal{E}_\sigma G^u_{\sigma, \Bar{\sigma}} & -\mathcal{E}_\sigma G^w_\sigma \end{bmatrix}.
\end{aligned}
\end{equation}
Fix {\tb $Q_f = Q_f^\transp > 0$ and $Q_s = Q_s^\transp > 0$}. Let ${\tb P_f = P_f^\transp > 0}$ be the unique solution to
\begin{equation}
    A^\transp P_f + P_f A = -Q_f.
\end{equation}
 By Theorem \ref{Theorem:RobStabLyap},  there exists $M^\sigma > 0$ such that for any $\mathcal{E}_\sigma \in \mathcal{D}_{\rm blk}$, there exists ${\tb P_s = P_s^\transp > 0}$ such that $\|P_s \mathcal{E}_\sigma \|_2 \leq M^\sigma$, and 
\begin{equation}
    (-G^u_{\sigma, \sigma} K_\sigma)^\transp \mathcal{E}_\sigma P_s + P_s \mathcal{E}_\sigma (-G^u_{\sigma, \sigma} K_\sigma) = -Q_s.
\end{equation}
Consider the composite Lyapunov candidate 
\begin{equation}\label{Eq:LTILyap0}
    V_\sigma(\xi,\eta_{\sigma}) = (*)^\transp \mathcal{P} \begin{bmatrix} \xi \\ \eta_\sigma - \bar{\eta}^\sigma_\sigma \end{bmatrix}, \,\, \mathcal{P} = \begin{bmatrix} P_f & 0 \\ 0 & \varepsilon_{\rm max} P_s \end{bmatrix},
\end{equation}
where $\varepsilon_{\rm max} \define \| \mathcal{E}_\sigma \|_{2} $. Routine calculations show
\begin{equation}\label{Eq:LTILyap}
\begin{aligned}
\mathcal{A}^\transp \mathcal{P} + \mathcal{P} \mathcal{A} &= -\begin{bmatrix} Q_f - X^\sigma & -{Y^\sigma}^\transp \\
    -Y^\sigma & \varepsilon_{\rm max} Q_s \end{bmatrix} \define -\mathcal{Q},
\end{aligned}
\end{equation}
where
\begin{equation} \label{eq:XYDef}
\begin{aligned}
    X^\sigma &\triangleq (\Pi^u_\sigma K_\sigma \mathcal{E}_\sigma C_\sigma)^\transp P_f + P_f (\Pi^u_\sigma K_\sigma \mathcal{E}_\sigma C_\sigma), \\
    Y^\sigma &\triangleq (\Pi^u_\sigma K_\sigma \mathcal{E}_\sigma G^u_{\sigma, \sigma} K_\sigma)^\transp P_f + \varepsilon_{\rm max} P_s \mathcal{E}_\sigma (-C_\sigma).
\end{aligned}
\end{equation}
Note that $X^\sigma, Y^\sigma$ satisfy the bounds
\begin{equation*}
\begin{aligned}
    \|X^\sigma\| \leq c^\sigma_1\varepsilon_{\rm max},
\end{aligned}
\qquad 
\begin{aligned}
    \|Y^\sigma\| \leq c^\sigma_2\varepsilon_{\rm max}
\end{aligned}
\end{equation*}
for constants $c^\sigma_1, c^\sigma_2 > 0$ which are independent of $\mathcal{E}_\sigma$; in obtaining the bound for $Y^\sigma$, we have invoked the bound $\| P_s \mathcal{E}_\sigma \|_2 \leq M^\sigma $. Since ${\tb \varepsilon_{\rm max}Q_s > 0}$, it follows that ${\tb \mathcal{Q} > 0}$ if and only if \cite[Proposition 8.2.4]{DSB:11} {\tb
\begin{equation} \label{eq:FinalSchurP1LTI}
    Q_f - X^\sigma - {Y^\sigma}^T(\varepsilon_{\rm max}Q_s)^{-1}Y^\sigma > 0,
\end{equation} }
which holds if $\lambda_{\rm min}(Q_f) - c^\sigma_3\varepsilon_{\rm max} > 0$, where $c^\sigma_3 \define (c^\sigma_1+(c^\sigma_2)^2/\lambda_{\rm min}(Q_s))$. More simply, \eqref{eq:FinalSchurP1LTI} holds if $\varepsilon_{\rm max} \in (0,\varepsilon^{*, \sigma})$ where $\varepsilon^{*, \sigma} \define \lambda_{\rm min}(Q_f)/c^\sigma_3$. We define $\varepsilon^* \define \min_{\sigma \subset \indx{N}} \varepsilon^{*, \sigma}$ to be our overall upper tuning gain across the set of systems \eqref{eq:ClosedLoopSigma}. Thus, for any $\sigma \in \indx{N}$, if $\varepsilon_\sigma \in (0,\varepsilon^*)^{|\sigma|}$, the equilibrium point $(\bar{x}^\sigma, \bar{\eta}^\sigma_\sigma)$ will be exponentially stable.
\end{proof}

{\tb Analyzing \eqref{eq:XYDef} and the deductions below \eqref{eq:FinalSchurP1LTI}, one finds $\varepsilon^{*}$ shrinks as $M^\sigma$ (given by Theorem \ref{Theorem:RobStabLyap}) grows. As the dependence of $M^\sigma$ on the matrix $-G^uK$ is not well understood, the bound $\varepsilon^*$ remains qualitative. Characterizing the relationship between $M^\sigma$, $\mu$, and structural properties of $-G^uK$ is a topic for future study.}

\subsection{Numerical Example}

Consider an exponentially stable two partition LTI system of the form \eqref{eq:LTIExpanded}, with $x \in \real^3$, $u_1, e_1 \in \real^2$, $u_2, e_2 \in \real^1$, 
{\tb
\begin{equation*}
    A =
    \left[
    \begin{array}{c c c}
    -10 & 0 & 10 \\ -100 & -10 & 100 \\ 100 & 10 & -110
    \end{array}
    \right],
\end{equation*} 
}$B^u = C = I \in \real^{3 \times 3}$, and $D^u = 0 \in \real^{3 \times 3}$. We additionally interpret $w \in \real^3$ as a constant reference we wish to track, leading to the selections $D^w = -I \in \real^{3 \times 3}$ and $B^w = 0 \in \real^{3 \times 3}$. We implement a decentralized low-gain integral controller of the form \eqref{eq:ControllersExpanded}. Selecting $K = I \in \real^{3 \times 3}$, we calculate

\begin{equation*}
    -G^u K =
    \left[
    \begin{array}{c c|c}
    -0.1 & -0.1 & -0.1 \\ 1 & -0.1 & 0 \\ \hline 0 & -0.1 & -0.1
    \end{array}
    \right].
\end{equation*}

One can show $-G^u K$ is robustly $D$-stable (and thus trivially robustly block $D$-stable) \cite[Condition 12]{WSK:02}. Alternatively, one may leverage Theorem \ref{Theorem:RobStabLyap} to achieve the same result (which has been omitted for space). Notably, one finds $-G^u K$ is \emph{not} block diagonally stable using standard convex optimization techniques (e.g., \cite{SB-QY:89}). Applying Theorem \ref{Theorem:ProbSol}, the closed loop system is exponentially stable for all  sufficiently small tunings of $(\varepsilon_1,\varepsilon_2)$. We simulate the system with three different sets of tuning parameters determined experimentally, with initial conditions $x(0) = \eta(0) = 0 \in \real^3$, and reference $w = \begin{bmatrix} 10 & 15 & 5 \end{bmatrix}^\transp$. The first and second control loops are closed at $t_{\mathrm{cl},1} = 0$ and $t_{\mathrm{cl},2} = 75$ respectively. Simulation results are shown in Figure \ref{fig:errorSignal}, which illustrate how the robust block $D$-stability condition allows for arbitrary loop closures and disconnections.

\begin{figure*}
    \centering
    \includegraphics[width=\textwidth]{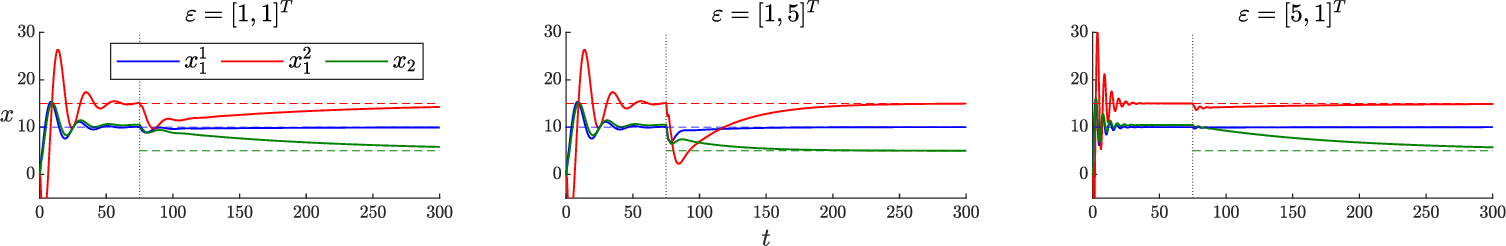}
    \caption{State reference tracking; vertical line indicates closure of the second control loop.}
    \label{fig:errorSignal}
\end{figure*}

\section{Conclusion}
\label{sec:conclusion}

We have developed a Lyapunov characterization for the robust block $D$-stability property, and demonstrated how it can be leveraged in a Lyapunov-based proof to establish stability conditions for systems under decentralized low-gain integral control. Future work will focus on extending the robust $D$-stability concept to nonlinear systems, with intended applications in nonlinear decentralized integral control and  decentralized feedback-based optimization.

\appendices

\section{Supporting Proofs} 
\label{Appendix:SupportProofs}

\begin{pfof}{Lemma \ref{Lemma:BlockDSeq}} 
    
    Recursively define the index sets $\Lambda_0 \define \emptyset $,
    \[
    \begin{aligned}
    \Lambda_\ell &\define \left\{ i \in \indx{N} \setminus \union_{r=0}^{\ell-1} \Lambda_r \, \Big| \limsup_{k \rightarrow \infty} \frac{d_j^k}{d_i^k} < \infty\right.\\
    &\qquad \qquad \qquad \qquad  \left.\bigg.\,\,\text{for all}\,\, j \in \indx{N} \setminus \union_{r=0}^{\ell-1} \Lambda_r \right\}
    \end{aligned}
    \]
    for $\ell \geq 1$. We define $\bar{N} \leq N$ as the smallest integer $\ell$ for which $\union_{r=1}^{\ell}\Lambda_{r} = \indx{N}$. By construction, $\Lambda_i \intersection \Lambda_j = \emptyset$ for all $i,j \in \indx{\bar{N}}$. For each $i \in \indx{\bar{N}}$, we now define the sequence elements $\tau_i^k \define d_{i^{\star}}^k$, where $i^{\star} \in \Lambda_{i}$ is a fixed arbitrary selection, and the sequence elements {\tb $\tilde{D}_i^k \define \mathrm{blkdiag} ( \{ d_j^k I / \tau_i^k \}_{j\in\Lambda_i} ) \in \mathcal{D}_{\rm blk}$.}
    Note that, by construction, $\lim_{k \rightarrow \infty} \tau_i^k / \tau_j^k = 0$ if and only if $i > j$. Additionally, for each $\ell \in \indx{\Bar{N}}$, we define the constant {\tb $c_{\ell} \define \sup_{k \in \natrl, i,j \in \Lambda_\ell} d_j^k / d_i^k < \infty$}
    and associated compact set {\tb $\mathcal{C}_\ell \define \{ D \in \mathcal{D}_\mathrm{blk} \mid c_\ell^{-1} I \leq D \leq c_\ell I \} \subset \mathcal{D}_\mathrm{blk}$ }. By construction, $\{\tilde{D}^k\}_{k=1}^{\infty} \subset \mathcal{C} \define \mathcal{C}_1 \times \ldots \times \mathcal{C}_{\Bar{N}}$. Existence of a unique block permutation matrix relating $D^k$ and the product $\Tilde{D}^k T^k$ is clear, which establishes the claim.
\end{pfof}

\section*{Acknowledgment}

The authors thank N. Monshizadeh for helpful literature pointers and discussions. The authors have used GPT-5.4 to assist with structural editing and grammar.

\bibliographystyle{IEEEtran}
\bibliography{brevalias, references}

\end{document}